\documentclass[11pt]{article}

\usepackage{xspace}
\usepackage{verbatim}
\usepackage{color}
\usepackage{graphicx}
\usepackage{tabularx}

\usepackage{amsmath}
\usepackage[showonlyrefs]{mathtools}
\usepackage{mathstyle}
\usepackage{breqn}
\usepackage{empheq}
\usepackage{amstext,amssymb,amsfonts}
\usepackage{fullpage}
\usepackage{nicefrac}
\usepackage{bm}

\usepackage{mdwlist}

\usepackage{todonotes}

\newcounter{mynotes}
\usepackage{nameref}
\usepackage[pagebackref,colorlinks,linkcolor=blue,filecolor = blue,
citecolor = blue, urlcolor = blue,pdfstartview=FitH]{hyperref}
\usepackage[capitalize]{cleveref}

\usepackage{textcomp,setspace}

\usepackage{amsthm}
\usepackage{thmtools}

\declaretheorem[within=section]{theorem}
\declaretheorem[sibling=theorem]{corollary}

\declaretheorem[sibling=theorem]{lemma}
\declaretheorem[sibling=theorem]{claim}

\declaretheorem[sibling=theorem]{definition}
\declaretheorem[sibling=theorem]{proposition}

\crefname{conjecture}{Conjecture}{Conjectures}
\crefname{claim}{Claim}{Claims}
\crefname{remark}{Remark}{Remarks}
\crefname{Lemma+Definition}{Lemma+Definition}{Lemma+Definition}

\newcounter{termcounter}
\renewcommand{\thetermcounter}{\Alph{termcounter}}
\crefname{term}{term}{terms}
\creflabelformat{term}{\textup{(#2#1#3)}}

\makeatletter
\def\term{\@ifnextchar[\term@optarg\term@noarg}
\def\term@optarg[#1]#2{%
  \textup{(#1)}%
  \def\@currentlabel{#1}%
  \def\cref@currentlabel{[][2147483647][]#1}%
  \cref@label[term]{#2}}
\def\term@noarg#1{%
  \refstepcounter{termcounter}%
  \textup{(\thetermcounter)}%
  \cref@label[term]{#1}}
\makeatother

\newcommand{\msf}[1]{\mathsf {#1}}

\newcommand{\ignore}[1]{}

\newcommand{\dist}{\mathrm{dist}}

\newcommand{\Brac}[1]{\left[#1 \right]}
\newcommand{\set}[1]{\left\{#1\right\}}

\definecolor{DSred}{rgb}{1,0,0}

\renewcommand{\leq}{\leqslant}
\renewcommand{\geq}{\geqslant}
\renewcommand{\ge}{\geqslant}
\renewcommand{\le}{\leqslant}
\renewcommand{\epsilon}{\varepsilon}
\newcommand{\eps}{\epsilon}

\newcommand{\R}{\mathbb{R}}
\newcommand{\C}{\mathbb{C}}
\newcommand{\Z}{\mathbb{Z}}
\newcommand{\N}{\mathbb{N}}
\newcommand{\F}{\mathbb{F}}
\newcommand{\T}{\mathbb{T}}
\newcommand{\U}{\mathbb{U}}

\newcommand{\cB}{\mathcal B}

\newcommand{\cH}{\mathcal H}

\newcommand{\cP}{\mathcal P}

\newcommand{\cT}{\mathcal T}

\newcommand{\Esymb}{{\bf E}}

\newcommand{\Psymb}{{\bf Pr}}

\DeclareMathOperator*{\E}{\Esymb}
\DeclareMathOperator*{\Var}{\Vsymb}
\DeclareMathOperator*{\ProbOp}{\Psymb}

\renewcommand{\Pr}{\ProbOp}

\newcommand{\Ex}[1]{\E\Brac{#1}}

\newcommand{\expo}[1]{{\mathsf{e}\left(#1\right)}}

\title{Estimating the distance from testable affine-invariant properties}
\author{Hamed Hatami\\
McGill University\\
\texttt{hatami@cs.mcgill.ca}
\and
Shachar Lovett\\
UC San Diego\\
\texttt{slovett@cs.ucsd.edu}
}

\def\Var{\mathbf{Var}}

\begin{document}
\maketitle

\begin{abstract}
Let $\cP$ be an affine invariant property of functions $\F_p^n \to [R]$ for fixed $p$ and $R$. We show that if $\cP$ is locally testable with a constant number of queries, then one can estimate the distance of a function $f$ from $\cP$ with a constant number of queries. This was previously unknown even for simple properties such as cubic polynomials over $\F_2$.

Our test is simple: take a restriction of $f$ to a constant dimensional affine subspace, and measure its distance from $\cP$. We show that by choosing the dimension large enough, this approximates with high probability the global distance of $f$ from $\cP$. The analysis combines the approach of Fischer and Newman [SIAM J. Comp 2007] who established a similar result for graph properties, with recently developed tools in higher order Fourier analysis, in particular those developed in Bhattacharyya et al. [STOC 2013].
\end{abstract}

\section{Introduction}\label{sec:intro}

Blum, Luby, and Rubinfeld~\cite{BLR} observed that given a function $f:\F_p^n \rightarrow \F_p$, it is possible to inquire the value of $f$ on a few random points, and accordingly probabilistically distinguish between the case that $f$ is a linear function and the case that $f$ has to be modified on at least $\eps>0$ fraction of points to become a linear function. Inspired by this observation, Rubinfeld and Sudan~\cite{866666} defined the concept of property testing which is now a major area of research in theoretical computer science. Roughly speaking, to test a function for a property means to examine the value of the function on a few random points, and accordingly (probabilistically) distinguish between
the case that the function has the property and the case that it is not too close to any function with that property.

The focus of our work is on testing properties of multivariate functions over finite fields. Fix a prime $p \geq 2$ and an integer $R \geq 2$ throughout. Let $\F = \F_p$ be a prime field and $[R]=\{0,\ldots,R-1\}$. We consider properties of functions $f : \F^n \to [R]$.
We are interested in testing the \emph{distance} of a function $f:\F^n \to [R]$ to a property. Here the distance corresponds to the minimum fraction of the points on which the function can be modified in order to satisfy the property. Fischer and Newman~\cite{MR2318716} showed that it is possible to estimate the distance from a graph to any given \emph{testable} graph property. In this article we extend this result to the algebraic setting of affine-invariant properties on functions $f:\F^n \to [R]$. Furthermore we show that the Fischer-Newman test can be replaced by a more natural one: pick a sufficiently large subgraph $H$ randomly and estimate the distance of $H$ to the property. Analogously, in our setting, we pick a sufficiently large affine subspace of $\F^n$ randomly, and measure the distance of the restriction of the function to this subspace from the property.

\subsection{Testability}
Given a property $\cP$ of functions in $\{\F^n \to [R] \ | \ n \in \N\}$,
we say that $f : \F^n \to [R]$ is {\em $\eps$-far} from $\cP$ if
$$\min_{g \in \cP} \Pr_{x \in \F^n}[f(x) \neq g(x)] > \eps,$$
and we say that it is {\em $\eps$-close} otherwise.

\begin{definition}[Testability]\label{testable}
A property $\cP$ is said to be {\em testable} (with two-sided error)
if there is a function $q: (0,1) \to \N$ and an algorithm $T$ that, given as input a parameter $\eps > 0$ and oracle
access to a function $f: \F^n \to [R]$, makes at most $q(\eps)$
queries to the oracle for $f$, accepts with probability at least $2/3$ if $f \in \cP$ and
rejects with probability at least $2/3$ if $f$ is $\eps$-far
from $\cP$.
\end{definition}

Note that if we do not require any restrictions on $\cP$, then the algebraic structure of $\F^n$ becomes irrelevant, and $\F^n$ would be treated as a generic set of size $|\F|^n$. To take the algebraic structure into account, we have to require certain ``invariance'' conditions.

We say that a property $\cP \subseteq \{\F^n \to [R]\ | \ n \in \N\}$ is \emph{affine-invariant} if for any $f \in \cP$ and  any affine transformation $A:
\F^n \to \F^n$, we have $Af  :=f \circ A \in \cP$ (an affine transformation
$A$ is of the form $L+c$ where $L$ is linear and $c$ is a constant vector in $\F^n$).
Some well-studied examples of affine-invariant properties
include Reed-Muller codes (in other words, bounded degree
polynomials) \cite{BFL, BFLS,FGLSS,RS,AKKLR} and Fourier sparsity
\cite{GOSSW}. In fact, affine invariance seems to be a common feature of most interesting
properties that one would classify as ``algebraic''.  Kaufman and
Sudan in \cite{KS08} made explicit note of this phenomenon and initiated a general study of the testability of
affine-invariant properties (see also~\cite{MR2863292}). In particular, they asked for necessary
and sufficient conditions for the testability of affine-invariant properties. This question initiated
an active line of research, which have led to a near complete characterization of testable affine invariant properties,
at least in the regime of one-sided error~\cite{BCSX09, KSV12, Sha09, BGS10,BFL12,BFHHL13}.

It is not difficult to see that for affine-invariant properties testability has an equivalent ``non-algorithmic'' definition through the distribution of restrictions to affine subspaces. We will describe a restriction of $\F^n$ to an affine subspace of dimension $k$ by an affine embedding $A:\F^k \to \F^n$
(an affine embedding is an injective affine transformation). The restriction of $f:\F^n \to [R]$ to the subspace is then given by $Af:\F^k \to [R]$.

\begin{proposition}\label{prop:testing}
An affine-invariant property $\cP$ is testable if and only if for every $\eps>0$, there exist a constant $k$ and a set $\cH \subseteq \{\F^k \to [R]\}$, such that for a function $f:\F^n \to [R]$ and a random affine embedding $A:\F^k \to \F^n$ the following holds. If $f \in \cP$, then
$$\Pr[Af \in \cH]>2/3,$$
and if $f$ is $\eps$-far from $\cP$, then
$$\Pr[Af  \not\in \cH]>2/3.$$
\end{proposition}

\subsection{Our contribution}

For a property $\cP$ and a positive real $\delta$, let $\cP_\delta$ denote the set of all functions that are $\delta$-close to the property. Our main result is the following theorem.

\begin{theorem}\label{thm:main}
For every testable affine-invariant property $\cP$ and every $\delta>0$, the property $\cP_\delta$ is testable.
\end{theorem}

\cref{thm:main} says that for every $\eps,\delta>0$ one can probabilistically distinguish between functions that are $\delta$-close to the property and the functions that are $(\delta+\eps)$-far from the property using only a constant number of queries. In fact the test is very natural. We show that there exists a constant $k_{\eps,\delta,\cP}$ such that for a random affine embedding $A:\F^k \to \F^n$, with probability at least $2/3$, $\dist(Af, \cP)$ provides a  sufficiently accurate estimate of $\dist(f,\cP)$. Hence our test will be the following: Pick a random affine embedding $A:\F^k \to \F^n$. If $\dist(Af, \cP)< \delta+\frac{\eps}{2}$ accept, otherwise reject.
This corresponds to taking $\cH=\left\{h: \F^k \to [R] \ | \ \dist(h,\cP) \le \delta+\frac{\eps}{2}\right\}$ in \cref{prop:testing}.

We note that previously it was unknown if one can test distance to even simple properties, such as cubic polynomials over $\F_2$. The reason was that one specific natural test (the Gowers norm, or derivatives test) was shown not to perform well for such properties. Our work shows that a natural test indeed works, albeit the number of queries have to grow as a function of $\eps$. We do not know if this is necessary for simple properties, such as cubic polynomials over $\F_2$, and leave this as an open problem.

On a technical level, our work combines two technologies developed in previous works. The first is the work of Fischer and Newman~\cite{MR2318716} which obtained similar results for graph properties. The second is higher order Fourier analysis, in particular a recent strong equidistribution  theorem established in Bhattacharyya et al.~\cite{BFHHL13}. From a high level, the approach for the graph case and the affine-invariant case are similar. One applies a regularization process, which allows to represent a graph (or a function) by a small structure. Then, one argues that a large enough random sample of the graph or function should have a similar small structure representing it. Hence, properties of the main object can be approximated by properties of a large enough sample of it. Fischer and Newman~\cite{MR2318716} implemented this idea in the graph case. We follow a similar approach in the algebraic case, which inevitably introduces some new challenges. One may see this result as an outcome of the large body of work on higher-order Fourier analysis developed in recent years. Once the machinery was developed, we can now apply it in various frameworks which were not accessible previously.

\subsection{Proof overview}
Let $R=2$ for the simplicity of exposition, e.g. we consider functions $f:\F^n \to \{0,1\}$. Let $\cP$ be an affine invariant property of functions $\{\F^n \to \{0,1\}: n \in \N\}$ which is locally testable, and fix $\eps,\delta>0$. We want to show that there exists an $m$ (which depends only on $\cP,\eps,\delta$) such that the following holds. Let $f:\F^n \to \{0,1\}$ be a function, and let $\tilde{f}$ be the restriction of the function to a random $m$-dimensional affine subspace of $\F^n$. Then
\begin{itemize}
\item Completeness: If $f$ is $\delta$-close to $\cP$ then, with high probability, $\tilde{f}$ is $(\delta+\eps/2)$-close to $\cP$.
\item Soundness: If $f$ is $(\delta+\eps)$-far from $\cP$ then, with high probability, $\tilde{f}$ is $(\delta+\eps/2)$-far from $\cP$.
\end{itemize}
Once we show that we are done, as the local test computes the distance of $\tilde{f}$ from $\cP$. If it is below $\delta+\eps/2$ we declare that $f$ is $\delta$-close to $\cP$; otherwise we declare it is $(\delta+\eps)$-far from $\cP$. The test correctness follows immediately from the completeness and soundness. We next argue why these hold.

Let us first fix notations. Let $A:\F^m \to \F^n$ be a random full rank affine transformation. Then, a restriction of $f$ to a random $m$-dimensional affine subspace can be equivalently described by $\tilde{f}=Af$. The proof of the completeness is simple. If $f$ is $\delta$-close to a function $g:\F^n \to \{0,1\}$ which is in $\cP$, then with high probability over a random restriction, the distance of $Af$ and $Ag$ is also at most $\delta+o_m(1)$. This is true because a random affine subspace is pairwise independent with regards to whether an element is contained in it. This, combined with Chebyshev's inequality implies the result. Then, by choosing $m$ large enough we get the error term down to $\eps/2$.

The main work (as in nearly all works in property testing) is to establish soundness. That is, we wish to show that if a function $f$ is far from $\cP$ then, with high probability, a random restriction of it is also from from the property. The main idea is to show that if for a typical restriction $Af$ is $\delta$-close to a function $h:\F^m \to \{0,1\}$ which is in $\cP$, then $h$ can be ``pulled back'' to a function $g:\F^n \to \{0,1\}$ which is both roughly $\delta$-close to $f$ and also very close to $\cP$. This will contradict our initial assumption that $f$ is $(\delta+\eps)$-far from $\cP$. In order to do so we apply the machinery of higher order Fourier analysis. The first description will hide various ``cheats'' but will present the correct general outline. We then note which steps need to be fixed to make this argument actually work.

First, we apply the assumption that $\cP$ is locally testable to derive there exist a constant dimension $k=k(\cP,\eps)$ so that a random restriction to a $k$-dimensional subspace can distinguish functions in $\cP$ from functions which are $\eps/4$-far from $\cP$. We want to decompose $f$ to ``structured'' parts which we will study, and ``pseudo-random'' parts which do not affect the distribution of restrictions to $k$-dimensional subspaces. In order to do so, for a function $f:\F^n \to \{0,1\}$ define by $\mu_{f,k}$ the distribution of its restriction to $k$-dimensional subspaces. That is, for $v:\F^k \to \{0,1\}$ let
$$
\mu_{f,k}[v] = \Pr_A[Af=v].
$$
We need to slightly generalize this definition to functions where the output $f(x)$ can be random. In our context, a randomized function is a function $f:\F^n \to [0,1]$, which describes a distribution over functions $f':\F^n \to \{0,1\}$, where for all $x$ independently $\Pr[f'(x)=1]=f(x)$. We extend the definition of $\mu_{f,k}$ to randomized functions by $\mu_{f,k}[v] = \E_{A,f'} \mu_{f',k}[v]$. By our definition, if two functions $f,g:\F^n \to [0,1]$ have distributions $\mu_{f,k}$ and $\mu_{g,k}$ close in statistical distance, then random restrictions to $k$-dimensional affine subspaces cannot distinguish $f$ from $g$. This will be useful in the analysis of the soundness.

We next decompose our function $f$ based on the above intuition. The formal notion of pseud-randomness we use is that of Gowers uniformity. Informally, the $d$-th Gowers uniformity measures correlation with polynomials of degree less than $d$. However, it turns to capture much more than that. For example, one can show that by choosing $d$ large enough ($d=p^k$ suffices) then for any functions $f,g:\F^n \to [0,1]$, if $\|f-g\|_{U^d}$ is small enough then $\mu_{f,k}$ and $\mu_{g,k}$ are close in statistical distance. Thus, it makes sense to approximate $f$ as
$$
f=f_1 + f_2
$$
where $f_1$ is structured (to be explained soon) and $\|f_2\|_{U^d}$ is small enough. This will allow us to replace $f$ with $f_1$ for the purposes of analyzing its restrictions to $k$-dimensional subspaces. The structure of $f_1$ is as follows: it is a function of a constant number $C=C(\cP,\eps)$ of polynomials of degree less than $d$. That is,
$$
f_1(x)=\Gamma(P_1(x),\ldots,P_C(x)),
$$
where $P_1,\ldots,P_C$ are polynomials and $\Gamma:\F^C \to \{0,1\}$ is some function (not necessarily a low degree polynomial). The benefit of this decomposition is that $f_1$ is ``dimension-less'' in the sense that $\Gamma$ does not depend on $n$; however, the polynomials $P_1,\ldots,P_C$ do depend on $n$. One can however ``regularize'' these polynomials in order to obtain ``random-looking'' (or high rank) polynomials. It can be shown that all properties of high rank polynomials are governed just by their degree (which is at most $d$), hence essentially the entire description of $f_1$ does not depend on $n$.

The next step is to show that the same type of decomposition can be applied to the restriction $Af$ of $f$. Clearly, $Af=Af_1 + Af_2$. We show that with high probability over the choice of $A$,
\begin{itemize}
\item $Af_1=\Gamma(Q_1(x),\ldots,Q_C(x))$ where $Q_i=AP_i$ are the restrictions of $P_1,\ldots,P_C$; and
$Q_1,\ldots,Q_C$ are still of ``high enough rank'' to behave like random polynomials.
\item $\|Af_2\|_{U^d} \approx \|f_2\|_{U^d}$ so we can still approximate $Af \approx Af_1$ with respect to the distribution of their restrictions to random $k$-dimensional subspaces.
\end{itemize}

We next apply the same decomposition process to $h$, which we recall is the assumed function (in $m$ variables) which is $(\delta+\eps/2)$-close to $Af$. By choosing the conditions of regularity of $h$ slightly weaker than those of $f$ (but still strong enough), we get that we can decompose
$$
h = h_1 + h_2
$$
where
$$
h_1(x)=\Gamma'(Q_1(x),\ldots,Q_{C'}(x))
$$
for some $C'>C$ and $\|h_2\|_{U^d}$ is very small. The important aspect here is that, we can approximate $h$ by the structured function $h_1$, and moreover that the polynomials $Q_1,\ldots,Q_C$ which compose $Af_1$ are part of  the description of $h_1$. That is, both $Af_1$ and $h_1$ can be defined in terms of the same basic building blocks (high rank polynomials $Q_1,\ldots,Q_C$).

The next step is to ``pull back'' $h$ to a function defined on $\F^n$. An easy first step is to pull back $h_1$. We need to define for $C<i \le C'$ pullback polynomials $P_i:\F^n \to \{0,1\}$ of $Q_i:\F^n \to \{0,1\}$ such that both $Q_i = A P_i$; and such that $P_1,\ldots,P_{C'}$ are of high rank. This can be done for example by letting $P_i = D Q_i$ for any affine map $D:\F^n \to \F^m$ for which $AD$ is the identity map on $\F^m$. This provides a pull-back $\phi$ of the ``coarse'' description of $f_1$ of $h_1$, but does not in general generate a function close to $f$ (it makes sense, since we still haven't used the finer ``pseudo-random'' structure of $f$). Formally, we set $\phi(x)=\Gamma'(P_1(x),\ldots,P_{C'}(x))$. However, we can already show something about $\phi$: it is very close to $\cP$. More concretely, its distribution over restrictions to $d$-dimensional subspaces is very close to that of $h$. Hence, the tester which distinguishes function in $\cP$ from those $(\eps/4)$-far from $\cP$ cannot distinguish $\phi$ from functions in $\cP$, hence $\phi$ must be $(\eps/4)$-close to $\cP$.

The next step is to define a more refined pull-back of $f$. Define an atom as a subset $\{x \in \F^n:P_1(x)=a_1,\ldots,P_{C'}(x)=a_{C'}\}$ for values $a_1,\ldots,a_{C'} \in \F$. Note that the functions $f_1,h_1$ are constant over atoms. We next define $\psi:\F^n \to [0,1]$ by redefining $\phi$ inside each atom, so that the average over the atoms of $\phi,\psi$ is the same, but such that $\psi$ is as close as possible to $f$ given this constraint. For example, if in an atom the average of $f$ is higher than the value $\phi$ assigns to this atom (and so it needs to be reduced to match $\phi$), we set for all $x$ in this atom $\psi(x)=0$ if $\phi(x)=0$ and $\psi(x)=\alpha$ if $f(x)=1$, where $\alpha$ is appropriately chosen so that the averages match. We then show that $\psi$ is a proper pull-back of $h$ in the sense that
\begin{itemize}
\item The distance between $f,\psi$ is very close to the distance between $Af,h$, which we recall is at most $\delta+\eps/2$.
\item $\psi$ is nearly $\eps+4$ close to $\cP$ in the distributional sense.
\end{itemize}
To finalize, we show that sampling a function $g:\F^n \to \{0,1\}$ based on $\psi$ has the same properties, which shows that $f$ is not $(\delta+\eps)$-far from $\cP$.

Let us remark on a few technical points overlooked in the above description. First, there are the exact notions of ``high rank polynomials''. It turns that in order to make this entire argument work, one needs to consider
more general objects, called non-classical polynomials. We rely on a series of results on the distributional properties of high-rank non-classical polynomials, in particular these recently established in~\cite{BFHHL13}. Also, the decomposition theorems are actually to three parts,
$$
f=f_1+f_2+f_3,
$$
where $f_1$ is structured as before, $\|f_2\|_2$ is somewhat small (but not very small) and $\|f_3\|_{U^d}$ is very small. This requires a somewhat more refined analysis to make the argument work, but does not create any significant change in the proof outline as described above.

\subsection{Comparison with graph property testing}

The main outline of our proof follows closely that of Fischer and Newman~\cite{MR2318716}. They study graph properties, where decompositions are given by the Szemer\'{e}di regularity lemma. Their test, in the notation above, can be described as measuring the distance between $\Gamma$ and all potential $\Gamma'$ which can be achieved from graphs that have the property. Our argument (when applied to graph properties instead of affine invariant properties)  shows that a much more natural test achieves the same behaviour: choose a random small subgraph and measure its distance from the property. In quantitative terms it is hard to compare the two results, as both get outrageous bounds coming from the bounds in the regularity lemma. So, we view this part of our work as having contribution in the simplicity of the test, and not in terms of the simplicity of the proof or the quantitative bounds (which are both very similar).

The more challenging aspect of our work is to take this approach and carry it out in the affine invariant settings. The main reason is that in the affine invariant setup the structural parts have more structure in them than in the graph setting. In the graph setup, the structure of a graph can be represented by a constant size graph with weighted edges. In the affine invariant case, the structured part is a constant size function applied to polynomials. However there will be \emph{no constant bound} on  the number of variables, and they can grow as $n$ grows. So, at first glance, these ``compact descriptions'' have sizes which grow with the input size; this is very different from the graph case. The reason these compact descriptions are useful is because, as long as the polynomials participating in them are ``random enough'', then their exact definitions do  not matter, just a few simple properties of them (their degree, and ``depth'' for non-classical polynomials). This is fueled by the recent advances on higher-order Fourier analysis. In essence, the state of the art has reached a stage where these tools are powerful enough to simulate the counterpart arguments which were initially developed in the context of graph properties. Still, the affine invariant case is more complex, and there are some problems which we do not know yet how to handle. For example,
\begin{itemize}
\item A complete classification of one-sided testable properties (e.g. can properties of ``infinite complexity'' be locally testable?); See~\cite{BFHHL13}
\item A complete classification of two-sided testable properties.
\item Properties where any non-trivial distance from them can be witnessed by a constant number of queries (also called correlation testing~\cite{HL11}). For example, can one
test correlation to cubics over $\F_2$ using a constant number of queries?
\end{itemize}

\section{Background}

We need to recall some definitions and results about  higher order Fourier analysis. Most of the material in this section is directly quoted from the full version of~\cite{BFHHL13}.

\paragraph{Notation}
We shorthand $\F=\F_p$ for a prime finite field. For $f:\F^n \to \C$ we denote $\|f\|_1 = \E[|f(x)|]$, $\|f\|_2^2 = \E[|f(x)|^2]$ where $x \in \F^n$ is chosen uniformly and $\|f\|_{\infty} = \max |f(x)|$. Note that $\|f\|_1 \le \|f\|_2 \le \|f\|_{\infty}$. The expression $o_{m}(1)$ denotes quantities which approach zero as $m$ grows. We shorthand $x \pm \eps$ for any quantity in $[x-\eps, x+\eps]$.

\subsection{Uniformity norms and non-classical polynomials}

\begin{definition}[Multiplicative Derivative]\label{multderiv}
Given a function $f: \F^n \to \C$ and an element $h \in \F^n$, define
the {\em multiplicative derivative in direction $h$} of $f$ to be the
function $\Delta_hf: \F^n \to \C$ satisfying $\Delta_hf(x) =
f(x+h)\overline{f(x)}$ for all $x \in \F^n$.
\end{definition}

The \emph{Gowers norm} of order $d$ for a function $f: \F^n \to \C$ is the
expected multiplicative derivative of $f$ in $d$ random directions at
a random point.

\begin{definition}[Gowers norm]\label{gowers}
Given a function $f: \F^n \to \C$ and an integer $d \geq 1$, the {\em
  Gowers norm of order $d$}  for $f$ is given by
$$\|f\|_{U^d} = \left|\E_{y_1,\dots,y_d,x \in \F^n} \left[(\Delta_{y_1}
\Delta_{y_2} \cdots \Delta_{y_d}f)(x)\right]\right|^{1/2^d}.$$
\end{definition}
Note that as $\|f\|_{U^1}= |\Ex{f}|$ the Gowers norm of order $1$ is only a semi-norm. However for $d>1$, it is not difficult to show that
$\|\cdot\|_{U^d}$ is indeed a norm.

If $f = e^{2\pi i P/p}$ where $P: \F^n \to \F$ is a polynomial
of degree $< d$, then $\|f\|_{U^d} = 1$. If $d < p$ and $\|f\|_\infty \le 1$, then in fact, the
converse holds, meaning that any  function $f: \F^n \to \C$ satisfying
$\|f\|_\infty \le 1$ and $\|f\|_{U^d} = 1$ is of this form. But when $d \geq p$, the converse
is no longer true. In order to characterize functions $f : \F^n \to
 \C$ with $\|f\|_\infty \le 1$ and $\|f\|_{U^d}=1$, we define the notion of {\em
  non-classical polynomials}.

Non-classical polynomials might not be necessarily $\F$-valued. We need to
introduce some notation.
Let $\T$ denote the circle group $\R/\Z$. This is an abelian group
with group operation denoted $+$. For an integer $k \geq 0$, let $\U_k$
denote $\frac{1}{p^k} \Z/\Z$, a subgroup of $\T$.
Let $\iota: \F \to \U_1$ be the injection $x \mapsto \frac{|x|}{p} \mod 1$, where $|x|$ is the standard map from $\F$ to
$\set{0,1,\dots,p-1}$. Let $\msf{e}: \T \to \C$ denote the character
$\expo{x} = e^{2\pi i x}$.
\begin{definition}[Additive Derivative]\label{addderiv}
Given a function\footnote{We try to adhere to the following convention: upper-case letters (e.g. $F$ and
  $P$) to denote functions mapping from $\F^n$ to $\T$ or to $\F$,
  lower-case   letters (e.g. $f$ and $g$) to denote functions mapping
  from $\F^n$ to $\C$, and upper-case Greek letters (e.g. $\Gamma$ and
$\Sigma$) to denote functions mapping  $\T^C$ to $\T$. By  abuse of notation, we sometimes conflate $\F$ and $\iota(\F)$.} $P:
\F^n \to \T$ and an element $h \in \F^n$, define
the {\em additive derivative in direction $h$} of $f$ to be the
function $D_hP: \F^n \to \T$ satisfying $D_hP(x) = P(x+h) - P(x)$
for all $x \in \F^n$.
\end{definition}
\begin{definition}[Non-classical polynomials]\label{poly}
For an integer $d \geq 0$, a function $P: \F^n \to \T$ is said to be a
{\em non-classical polynomial of degree $\leq d$} (or simply a
{\em polynomial of degree $\leq d$}) if for all $y_1,
\dots, y_{d+1}, x \in \F^n$, it holds that
\begin{equation}\label{eqn:poly}
(D_{y_1}\cdots D_{y_{d+1}} P)(x) = 0.
\end{equation}
The {\em degree} of $P$ is the smallest $d$ for which the above holds.
A function $P : \F^n \to \T$ is said to be a {\em classical polynomial of degree
$\leq d$} if it is a non-classical polynomial of degree $\leq d$
whose image is contained in $\iota(\F)$.
\end{definition}

It is a direct consequence that a function $f :
\F^n \to \C$ with $\|f\|_\infty \leq 1$ satisfies $\|f\|_{U^{d+1}} =
1$ if and only if $f = \expo{P}$ for a
(non-classical) polynomial $P: \F^n \to \T$ of degree $\leq d$.

\begin{lemma}[Lemma 1.7 in \cite{TZ11}]\label{struct}
A function $P: \F^n \to \T$ is a polynomial of degree $\leq d$ if and
only if $P$ can be represented as
$$P(x_1,\dots,x_n) = \alpha + \sum_{0\leq d_1,\dots,d_n< p; k \geq 0:
  \atop {0 < \sum_i d_i \leq d - k(p-1)}} \frac{ c_{d_1,\dots, d_n,
  k} |x_1|^{d_1}\cdots |x_n|^{d_n}}{p^{k+1}} \mod 1,
$$
for a unique choice of $c_{d_1,\dots,d_n,k} \in \set{0,1,\dots,p-1}$
and $\alpha \in \T$.  The element $\alpha$ is called the {\em
  shift} of $P$, and the largest integer $k$ such that there
exist $d_1,\dots,d_n$ for which $c_{d_1,\dots,d_n,k} \neq 0$ is called
the {\em depth} of $P$. Classical polynomials correspond to
polynomials with $0$ shift and $0$ depth.
\end{lemma}

Also, for convenience of exposition, we will assume throughout this
paper that the shifts of all polynomials are zero. This can be done
without affecting any of the results in this work. Hence, all
polynomials of depth $k$ take values in $\U_{k+1}$.

\subsection{Uniformity over linear forms}

A linear form in $k$ variables is $L=(\ell_1,\ldots,\ell_k) \in \F^k$. We interpret it as a linear operator $L:(\F^n)^k \to \F^n$ given by $L(x_1,\ldots,x_k) = \sum_{i=1}^k \ell_i x_i$.

\begin{definition}[Cauchy-Schwarz complexity, \cite{GT06}]\label{def:cplx}
Let $\mathcal{L} = \{L_1,\dots,L_m\}$ be a set of linear forms. The
{\em Cauchy-Schwarz complexity of $\mathcal{L}$} is the minimal $s$
such that the following holds. For every $i \in [m]$, we can partition
$\{L_j\}_{j \in [m]\setminus \{i\}}$ into $s+1$ subsets such that
$L_i$ does not belong to the linear span of each subset.
\end{definition}

Following is a lemma due to Green and Tao \cite{GT06} based on
repeated applications of the Cauchy-Schwarz inequality.
\begin{lemma}\label{lem:cnt}
Let $f_1, \dots, f_m : \F^n \to [-1,1]$. Let $\mathcal{L} =
\{L_1,\dots,L_m\}$ be a system of $m$ linear forms in $k$ variables
of Cauchy-Schwarz complexity $s$. Then:
$$\left|\E_{x_1,\dots,x_k \in \F^n} \left[ \prod_{i=1}^m
    f_i(L_i(x_1,\dots,x_\ell))\right] \right| \leq \min_{i \in [m]}\|f_i\|_{U^{s+1}}$$
\end{lemma}

We would need to apply \cref{lem:cnt} in this paper in the special case when $\mathcal{L}$ corresponds to all $p^k$ linear forms describing all points in an affine subspace of dimension $k$. We would care only about some upper bound on the Cauchy-Schwarz complexity of the system. The following claim follows immediately from the definitions.

\begin{claim}\label{claim:cs_kdim}
Let $\mathcal{L}=\{(1,a_1,\ldots,a_k): a_1,\ldots,a_k \in \F\}$. Then the Cauchy-Schwarz complexity of $\mathcal{L}$ is at most $p^k$.
\end{claim}

\subsection{Polynomial factors and rank}

\begin{definition}[Factors] If $X$ is a finite set then by a \emph{factor} $\cB$ we mean simply a
partition of $X$ into finitely many pieces called \emph{atoms}.
\end{definition}

A function $f:X \to \C$ is called \emph{$\cB$-measurable} if it is constant on atoms of $\cB$. For any function $f : X \to \C$, we may define
the conditional expectation
$$\E[f|\cB](x)=\E_{y \in \cB(x)}[f(y)],$$
where $\cB(x)$ is the unique atom in $\cB$ that contains $x$. Note that $\E[f|\cB]$ is $\cB$-measurable.

A finite collection of functions $\phi_1,\ldots,\phi_C$ from $X$ to some other finite space $Y$ naturally define a factor $\cB=\cB_{\phi_1,\ldots,\phi_C}$ whose atoms are sets of the form $\{x: (\phi_1(x),\ldots,\phi_C(x))= (y_1,\ldots,y_C) \}$ for some $(y_1,\ldots,y_C) \in Y^C$. By an abuse of notation
we also use $\cB$ to denote the map $x \mapsto (\phi_1(x),\ldots,\phi_C(x))$, thus also identifying the atom containing $x$ with
$(\phi_1(x),\ldots,\phi_C(x))$.

\begin{definition}[Polynomial factors]\label{factor}
If $P_1, \dots, P_C:\F^n \to \T$ is  a sequence of polynomials, then the factor $\cB_{P_1,\ldots,P_C}$ is called a {\em polynomial factor}.
\end{definition}

The {\em complexity} of $\cB$, denoted $|\cB|$, is the number of defining polynomials
$C$.  The {\em degree} of $\cB$ is the maximum
degree among its defining polynomials $P_1,\ldots,P_C$. If $P_1,\ldots,P_C$ are of depths $k_1,\ldots,k_C$, respectively,
then $\|\cB\|=\prod_{i=1}^C p^{k_i+1}$ is called the \emph{order} of $\cB$.

Notice that the number of atoms of $\cB$ is bounded by $\|\cB\|$.

Next we need to define the notion of the {\em rank} of a polynomial or a polynomial factor.

\begin{definition}[Rank of a polynomial]
Given a polynomial $P : \F^n \to \T$ and an integer $d > 1$, the {\em $d$-rank} of
$P$, denoted $\msf{rank}_d(P)$, is defined to be the smallest integer
$r$ such that there exist polynomials $Q_1,\dots,Q_r:\F^n \to \T$ of
degree $\leq d-1$ and a function $\Gamma: \T^r \to \T$ satisfying
$P(x) = \Gamma(Q_1(x),\dots, Q_r(x))$. If $d=1$, then
$1$-rank is defined to be $\infty$ if $P$ is non-constant and $0$
otherwise.

The {\em rank} of a polynomial $P: \F^n \to \T$ is its $\deg(P)$-rank.
\end{definition}

A high-rank polynomial of degree $d$ is, intuitively, a ``generic''
degree-$d$ polynomial. There are no unexpected ways to decompose it
into lower degree polynomials. The following theorem shows that a high rank polynomial is distributed close to uniform.

\begin{theorem}[Theorem 4 of \cite{KL08}]\label{rankreg}
For any $\eps > 0$ and integer $d > 0$, there exists
$r = r_{\ref{rankreg}}(d,\eps)$ such that the following is true.
If $P: \F^n \to \T$ is a degree-$d$ polynomial with rank greater than $r$,
then $|\E_{x}[\expo{P(x)}]| < \eps$.
\end{theorem}

 Next, we will formalize the
notion of a generic collection of polynomials. Intuitively, it should
mean that there are no unexpected algebraic dependencies among the
polynomials.

\begin{definition}[Rank and Regularity]\label{regular}
A polynomial factor $\cB$ defined by a sequence of
polynomials $P_1,\dots,P_C: \F^n \to \T$ with respective depths $k_1,
\dots, k_C$ is said to have {\em   rank
  $r$} if $r$ is the least integer for which there exist
$(\lambda_1,\dots, \lambda_C) \in \Z^C$ so that $(\lambda_1 \mod p^{k_1 + 1}, \dots, \lambda_C \mod
p^{k_C + 1}) \neq (0, \dots, 0)$
and the polynomial $Q = \sum_{i=1}^C \lambda_i P_i$ satisfies
$\msf{rank}_{d}(Q) \leq r$ where $d = \max_{i} \deg(\lambda_iP_i)$.

Given a polynomial factor $\cB$ and a function $r: \N
\to \N$, we say $\cB$ is {\em $r$-regular} if $\cB$ is of rank
larger than $ r(|\cB|)$.
\end{definition}

Note that since $\lambda$ can be a multiple of $p$, rank measured with
respect to $\deg(\lambda P)$ is not the same as rank measured with
respect to $\deg(P)$. So, for instance, if $\cB$ is the factor defined by a
single polynomial $P$ of degree $d$ and depth $k$, then
$$\msf{rank}(\cB) = \min\set{\msf{rank}_d(P), \msf{rank}_{d-(p-1)}(pP),
\cdots, \msf{rank}_{d-k(p-1)}(p^kP)}.$$

Regular factors indeed do behave like  generic collections of
polynomials, and thus, given any factor $\cB$ that is not regular, it
will often be useful to {\em regularize} $\cB$, that is, find a refinement
$\cB'$ of $\cB$ that is regular up to our desires. We distinguish
between two kinds of refinements.

\begin{definition}[Semantic and syntactic refinements] \label{refine}
A polynoial factor $\cB'$ is called a {\em syntactic refinement} of $\cB$, and
denoted $\cB' \succeq_{syn} \cB$, if the sequence of polynomials
defining $\cB'$ extends that of $\cB$. It is called a {\em
  semantic refinement}, and denoted $\cB' \succeq_{sem} \cB$ if the
induced partition is a combinatorial refinement of the partition
induced by $\cB$. In other words, if for every $x,y\in \F^n$,
$\cB'(x)=\cB'(y)$ implies $\cB(x)=\cB(y)$.
\end{definition}

The following lemma shows that every polynomial factor can be refined to be arbitrarily regular without increasing its complexity by more than a constant.

\begin{lemma}[Polynomial Regularity Lemma, Lemma 2.19 of~\cite{BFHHL13}]\label{factorreg}
Let $r: \N \to \N$ be a non-decreasing function and $d > 0$
be an integer. Then, there is a  function
$C_{\ref{factorreg}}^{(r,d)}: \N \to \N$  such
that the following is true. Suppose $\cB$ is a factor defined by
polynomials $P_1,\dots, P_C : \F^n \to \T$  of degree at most $d$.
Then, there is  an $r$-regular factor $\cB'$ consisting of  polynomials
$Q_1, \dots, Q_{C'}: \F^n \to \T$ of degree $\leq d$ such that $\cB'
\succeq_{sem} \cB$ and  $C' \leq  C_{\ref{factorreg}}^{(r,d)}(C)$.
\end{lemma}

The first step towards showing that regular factors behave like  generic collections of
polynomials is to show that they are almost equipartitions.

\begin{lemma}[Size of atoms, Lemma 3.2 of \cite{BFHHL13}]\label{atomsize}
Given $\eps > 0$, let $\cB$ be a polynomial factor of
degree $d > 0$,  complexity $C$, and rank
$r_{\ref{rankreg}}(d,\eps)$,   defined by a tuple of
polynomials $P_1, \dots, P_C: \F^n
\to \T$ having respective depths $k_1, \dots, k_C$.
Suppose $b = (b_1, \dots, b_C) \in \U_{k_1+1} \times \cdots \times \U_{k_C+1}$. Then
$$
\Pr_{x}[\cB(x) = b] =  \frac{1}{\|\cB\|} \pm \eps.
$$
In particular, for  $\eps < \frac{1}{\|\cB\|} $, $\cB(x)$ attains every possible value in its range and thus has
$\|\cB\|$ atoms.
\end{lemma}

%
%
%

Finally we state the regularity lemma, the basis of the higher order Fourier analysis.

\begin{theorem}[Theorem 4.4 of~\cite{BFL12}]\label{thm:reg}
Suppose $\zeta > 0$ is a real and $d,k \geq 1$ are integers. Let $\eta: \N \to \R^+$ be an arbitrary
non-increasing function, and let $r: \N \to \N$ be an arbitrary
non-decreasing function. Let $\cB_0$ be a polynomial factor of degree $d$ and complexity $C_0$. Then, there exist $C =
C_{\ref{thm:reg}}(\eta,  r, \zeta, C_0, d)$ such that the
following holds.

Every function $f: \F^n \to \{0,1\}$ has a decomposition $f=f_1+f_2+f_3$ such that the following is true:
\begin{itemize*}
\item $f_1 = \E[f| \cB_1]$ for a polynomial factor $\cB_1 \succeq_{sem} \cB_0 $ of degree
$d$ and complexity  $C_1 \le C$.
\item $\|f_2\|_2 < \zeta$ and $\|f_3\|_{U^{d+1}} < \eta(|\cB|)$.
\item The functions $f_1$ and $f_1+f_3$ have range
$[0,1]$; $f_2$ and $f_3$ have range $[-1,1]$.
\item $\cB_1$ is $r$-regular.
\end{itemize*}
Furthermore if $\msf{rank}(\cB_0) \ge r_{\ref{thm:reg}}(\eta,  r, \zeta, C_0, d)$, then one can assume that $\cB_1 \succeq_{syn} \cB_0$.
\end{theorem}

\subsection{Useful claims}
We prove in this subsection a few useful claims relating to polynomial factors.

Sometimes one needs to refine the decomposition given by the factor $\cB_1$ in \cref{thm:reg}. The following simple lemma is useful in such situations.
\begin{lemma}
\label{lem:refinement}
Let $\cB$ be a polynomial factor of complexity $C$ and degree $d$, and let $f:\F^n \to \{0,1\}$ be decomposed into $f=\E[f|\cB]+f_2+f_3$ for $f_2,f_3:\F^n \to [-1,1]$. If $\cB' \succeq_{sem} \cB$ is  a polynomial factor of degree $d$ and complexity $C'$, then
$$\left\|\Ex{f | \cB} - \Ex{f | \cB'}\right\|_1  \le \left\|f_2\right\|_{2} +  p^{dC'} \|f_3\|_{U^{d+1}}.$$
\end{lemma}
\begin{proof}
We have
$$\left\|\Ex{f | \cB} - \Ex{f | \cB'}\right\|_1  = \left\|\Ex{f_2 | \cB'} + \Ex{f_3 | \cB'}\right\|_1 \le  \left\|f_2\right\|_{1} + \left\| \Ex{f_3 | \cB'} \right\|_1.$$
The claim follows since $\left\|f_2\right\|_{1} \le \left\|f_2\right\|_{2}$ and \cref{cl:L2U} below shows that $\left\| \Ex{f_3 | \cB'} \right\|_1 \le  p^{dC'} \|f_3\|_{U^{d+1}}$.
\end{proof}

\begin{claim}
\label{cl:restrict_atom}
Let $f:\F^n \to [-1,1]$, and let $\cB$ be a polynomial factor of degree $d$ and complexity $C$.  Then for any atom $b$ of $\cB$,
$$\|f(x) 1_{\cB(x) = b}\|_{U^{d+1}} \le \|f\|_{U^{d+1}}.$$
\end{claim}
\begin{proof}
Let $\cB$ be defined by polynomials $P_1,\ldots,P_C$ of depths $k_1,\ldots,k_C$, respectively. An atom $b \in \cB$ is defined by $b = \{x \in \F^n: P_i(x)=b_i\}$. So
\begin{align*}
\|f(x) 1_{\cB(x) = b}\|_{U^{d+1}}
&= \left\|f(x)\prod_{i=1}^{C} \frac{1}{p^{k_i+1}} \sum_{
    \lambda_{i}=0}^{p^{k_i+1}-1} \expo{\lambda_{i} (P_{i}(x) - b_{i})}\right\|_{U^{d+1}}\\
&\le \prod_{i=1}^{C} p^{-(k_i+1)} \cdot \sum_{(\lambda_1, \dots, \lambda_C)
\atop \in \prod_{i} [0,p^{k_i + 1}-1]}
\left\|f(x)\expo{\sum_{i} \lambda_{i}(P_{i}(x) - b_{i})}\right\|_{U^{d+1}}\\
&= \left(\prod_{i=1}^C p^{-(k_i+1)}\right)  \cdot \left(\prod_{i=1}^C p^{k_i+1} \|f\|_{U^{d+1}}\right) = \|f\|_{U^{d+1}}.
\end{align*}Let us
\end{proof}

\begin{claim}
\label{cl:L2U}
Let $f:\F^n \to [-1,1]$, and let $\cB$ be a polynomial factor of degree $d$ and complexity $C$.  Then
$ \left\| \Ex{f | \cB} \right\|_1  \le p^{dC} \|f\|_{U^{d+1}}$.
\end{claim}
\begin{proof}
By the monotonicity of the Gowers norms and \cref{cl:restrict_atom}, for every atom $b \in \cB$
$$
\left|\Ex{f(x) 1_{\cB(x) = b}}\right| \le \left\|f(x) 1_{\cB(x) = b}\right\|_{U^{d+1}} \le \|f\|_{U^{d+1}}.
$$
Hence
$$\left\| \Ex{f | \cB} \right\|_1 = \sum_{b \in \cB} |\Ex{f(x) 1_{\cB(x)=b}}| \le  p^{dC} \|f\|_{U^{d+1}}.$$
\end{proof}

We need a simple bound on Gowers uniformity norms in terms of $L_1$ norm.

\begin{claim}\label{cl:gowers_l1}
Let $f:\F^n \to [-1,1]$. For any $d \ge 1$,
$$
\|f\|_{U^d} \le \|f\|_1^{1/2^d}.
$$
\end{claim}

\begin{proof}
By the definition of the $U^d$ norm and the boundedness of $f$,
$$
\|f\|_{U^d}^{2^d} \le \E_{x,y_1,\ldots,y_d} \prod_{I \subseteq [d]} |f(x + \sum_{i \in I} y_i)| \le \E |f(x)| = \|f\|_1.
$$
\end{proof}

We will also need the following lemma about restrictions of high rank polynomials to affine subspaces.

\begin{lemma}\label{lem:rankrestrict}
For $\eps>0$ and positive $d,e,r$, if $m \ge m_{\ref{lem:rankrestrict}}(\eps,d,e,r)$ and $n \ge m$ then the following holds. For every polynomial $P: \F^n \to \T$  of degree $d$, depth $e$ and rank $\geq r$, a random affine embedding $A:\F^m \to \F^n$ satisfies that
$$\Pr[\mbox{$\deg(P)<d$ or $(\msf{rank}_d(BP) <r)$ or $\mathrm{depth}(P)<e$}] < \eps.$$
\end{lemma}
\begin{proof}
Let $\cP'$ denote the property of ``bad'' restrictions of $P$. That is, $\cP'$ is the property of functions which are either polynomials of degree less than $d$; or polynomials of degree $d$ and rank less than $r$; or polynomials of degree $d$ and depth less than $e$. By assumption $P \notin \cP'$. Theorem 1.7 in~\cite{BFHHL13} shows that any degree-structural property, and in particular $\cP'$, is locally defined. Theorem 1.2 in~\cite{BFHHL13} shows that any such property is locally testable. Furthermore, as all elements in $\cP'$ are polynomials of degree $d$, then $P$ is $\eta$-far from $\cP'$, where $\eta \ge p^{-\lceil d/(p-1) \rceil}$ is the minimal distance of polynomials of degree $d$ (it is not hard to see that the minimal distance for non-classical polynomials of a given degree is achieved by a classical polynomial). Hence, by \cref{prop:testing} there exists $m=m(\eta,d,e,r)$ such that for a random affine embedding $A:\F^m \to \F^n$, $\Pr[AP \in \cP'] < \eps$.
\end{proof}

\section{Some remarks on testability}
Let us discuss some results related to testablity of affine-invariant properties. To simplify the presentation we focus on the special case of $R=2$. The proof easily generalizes to $R > 2$, by decomposing every function $f:\F^n \to [R]$ as $f=f^{(1)}+\ldots f^{(R)}$ where $f^{(i)}$ is the indicator function of the set $\{x : f(x)=i\}$.

Consider a function $f:\F^n \to \{0,1\}$, and a positive integer $k \le n$. Let $A:\F^k \to \F^n$ be a random  affine embedding. Let  $\mu_{f,k}$ denote the distribution of $Af:\F^k \to \{0,1\}$. So via this sampling, every function $f:\F^n \to \{0,1\}$ defines a probability measure $\mu_{f,k}$ on the set of functions $\{\F^k \to \{0,1\}\}$. We denote by $\mu_{f,k}[v]$ the probability that $\mu_{f,k}$ assigns to $v:\F^k \to \{0,1\}$.

This can be generalized to functions $f:\F^n \to [0,1]$. We view such functions as distribution over functions $f':\F^n \to \{0,1\}$, where $\Pr[f'(x)=1]=f(x)$ independently for all $x \in \F^n$. Let again $A:\F^k \to \F^n$ be a random  affine embedding, and we denote by $\mu_{f,k}$ the distribution of $Af':\F^k \to \{0,1\}$. This is a generalization of the former case as a function $f:\F^n \to \{0,1\}$ can be identified with the function that maps every $x \in \F^n$ to the point-mass probability distribution over $\{0,1\}$ which is concentrated on $f(x)$.

The following simple corollary follows easily from the definition of testability.
\begin{corollary}\label{cor:testing2}
If an affine-invariant property $\cP$ is testable, then for every $\eps>0$, there exist $k \ge 1$ so that the following holds. For any function $f:\F^n \to \{0,1\}$ with $\dist(f,\cP) \ge \eps$, and any function $g:\F^n \to \{0,1\}$ in $\cP$, the statistical distance between $\mu_{f,k}$ and $\mu_{g,k}$ is at least $1/3$.
\end{corollary}

\begin{proof}
From the definition of testability in \cref{prop:testing}, there exist $k \ge 1$ and a family $\cH \subseteq \{\F^k \to \{0,1\}\}$, such that for a random affine embedding $A:\F^k \to \F^n$, $\Pr[Ag \in \cH] > 2/3$ and $\Pr[Af  \in \cH] \le 1/3$. Hence, the statistical distance between $\mu_{f,k}$ and $\mu_{g,k}$ is at least $1/3$.
\end{proof}

We can deduce the following useful corollary. If $f:\F^n \to \{0,1\}$ has a distribution $\mu_{f,k}$ which is very close to $\mu_{g,k}$ for a function $g \in \cP$, then $f$ must be close to $\cP$. In fact, the same holds for $f:\F^n \to [0,1]$, except now the results holds with high probability over $f':\F^n \to \{0,1\}$ sampled from $f$.

\begin{corollary}\label{cor:testing3}
For $\eps>0$ let $k \ge 1$ be given by \cref{cor:testing2}, and assume that $n \ge n_{\ref{cor:testing3}}(k,\eps)$. Let $f:\F^n \to [0,1]$ and $g:\F^n \to \{0,1\}$ so that $g \in \cP$ and the statistical distance between $\mu_{f,k}, \mu_{g,k}$ is at most $1/4$. Let $f':\F^n \to \{0,1\}$ be sampled by taking $f'(x)=1$ with probability $f(x)$ independently for all $x \in \F^n$. Then with probability at least $99\%$ over the choice of $f'$,
$$
\dist(f',\cP) \le \eps.
$$
\end{corollary}

\begin{proof}
We will show that by choosing $n$ large enough, the distribution $\mu_{f,k}$ and $\mu_{f',k}$ are very close in statistical distance (say, distance $\le 1/12$) with high probability (say, $99\%$). The corollary then follows from \cref{cor:testing2} applied to $f'$ and $g$. Let $v:\F^k \to \{0,1\}$ be a function. By definition
$$
\mu_{f,k}[v] = \E_{A,f'} \Pr[Af'=v] = \E_{f'} \mu_{f',k}[v].
$$
Moreover, for two affine embeddings $A_1,A_2:\F^k \to \F^n$, if their images are disjoint then $A_1 f'$ and $A_2 f'$ are independent. Since the probability over a random choice of $A_1,A_2$ that their images intersect is at most $p^{2k-n}$, we get that
$$
\Var[\mu_{f',k}[v]] \le p^{2k-n} = o_{n}(1).
$$
This means that $\mu_{f',k}[v]=\mu_{f,k}[v]+o_n(1)$ with probability $1-o_n(1)$. The result now follows from applying the union bound over all possible values of $v$.
\end{proof}

We next argue that for two functions $f,g:\F^n \to [0,1]$, for any $k \ge 1$, there exists a $d \ge 1$ such that, if $\|f-g\|_{U^d}$ is small enough, then the statistical distance of $\mu_{f,k},\mu_{g,k}$ is arbitrarily small. This is useful, since it shows that in this case if $g \in \cP$ then $f$ must be close to $\cP$ provided that $k$ is  large enough.

\begin{lemma}\label{lemma:close_sampling_gowers}
For every $\eps>0, k \ge 1$, there exists $\rho>0, d \ge 1$ such that the following holds. If $f,g:\F^n \to [0,1]$ are functions such that $\|f-g\|_{U^{d}} \le \rho$, then the statistical distance between $\mu_{f,k}$ and $\mu_{g,k}$ is at most $\eps$.
\end{lemma}

\begin{proof}
Let $A:\F^k \to \F^n$ be a random affine embedding. Consider $v:\F^k \to \{0,1\}$. For $y \in \F^k$, define $f_y(x) = 1 - v(y) - (-1)^{v(y)} f(x)$. The probability that $\mu_{f,k}$ samples $v$ is given by
$$
\mu_{f,k}[v] = \E_{A,f'} \Pr[Af'=v] = \E_{A} \prod_{y \in \F^k} f_y(Ay).
$$
Similarly define $g_y(x) = 1 - v(y) - (-1)^{v(y)} g(x)$ to obtain
$$
\mu_{g,k}[v] = \E_{A} \prod_{y \in \F^k} g_y(Ay).
$$
Let $<$ define an arbitrary ordering on $\F^k$. We rewrite $\mu_{f,k}[v]-\mu_{g,k}[v]$ as a telescopic sum
$$
\mu_{f,k}[v] - \mu_{g,k}[v]= \sum_{z \in \F^k} \E_{A} \left(\prod_{y<z} f_y(Ay)\right) \cdot \left( f_z(Az)-g_z(Az) \right) \cdot \left( \prod_{y>z} g_y(Ay) \right).
$$
We bound each term in the sum. To do so, we will apply \cref{lem:cnt}. Note that the set of linear forms $\{Ay: y \in \F^k\}$ is exactly that given in \cref{claim:cs_kdim} and its Cauchy-Schwarz complexity is at most $p^k$. Note that $\|f_y\|_{\infty} \le 1$. Hence for $d=p^k+1$ we get that
$$
\left| \E_{A} \left(\prod_{y<z} f_y(Ay)\right) \cdot \left( f_z(Az)-g_z(Az) \right) \cdot \left( \prod_{y>z} g_y(Ay) \right)\right| \le \|f_z-g_z\|_{U^{d}} = \|f-g\|_{U^d}.
$$
We conclude that $|\mu_{f,k}[v] - \mu_{g,k}[v]| \le p^k \|f-g\|_{U^d}$ and hence the statistical distance between $\mu_{f,k}$ and $\mu_{g,k}$ is bounded by $2^{p^k}p^{k} \|f-g\|_{U^d}$. The lemma follows for $\rho=2^{-p^k}p^{-k} \eps$.
\end{proof}

The following corollary is immediate.

\begin{corollary}\label{cor:close_to_cP_gowers}
For every $\eps>0$ there exist $d \ge 1, \rho>0$ such that the following holds. Let $f,g:\F^n \to \{0,1\}$ be functions and assume that $g \in \cP$. If $\|f-g\|_{U^d} \le \rho$ then $f$ is $\eps$-close to $\cP$.
\end{corollary}

Structured parts obtained from the decomposition theorems are of the form $f(x) = \Gamma(P_1(x),\ldots,P_C(x))$ where $P_1,\ldots,P_C$ are polynomials. We would argue that if they have large enough rank, then $\mu_{f,k}$ essentially depends just on $\Gamma$ and the degrees and depths of $P_1,\ldots,P_C$, and not on the specific polynomials.

\begin{lemma}\label{lemma:mu_regular}
For any $\eps>0$ and $k, d \ge 1$, there exists $r=r_{\ref{lemma:mu_regular}}(k,d,\eps):\N \to \N$ such the following holds.
Let $P_1,\ldots,P_C$ be an $r$-regular factor over $\F^n$ of degree at most $d$. Let $Q_1,\ldots,Q_C$ be an $r$-regular factor over $\F^m$ of degree at most $d$. Assume that both $P_i,Q_i$ have degree $d_i \le d$ and depth $k_i$, for all $i \le C$. Let $\Gamma:\prod_{i=1}^C \U_{k_i+1} \to [0,1]$ be a function. Let $f:\F^n \to [0,1]$ be defined as $f(x)=\Gamma(P_1(x),\ldots,P_C(x))$ and $g:\F^m \to [0,1]$ be defined as $g(x)=\Gamma(Q_1(x),\ldots,Q_C(x))$. Then $\mu_{f,k}$ and $\mu_{g,k}$ have statistical distance at most $\eps$.
\end{lemma}

\begin{proof}
Let $A:\F^k \to \F^n$ be a random affine embedding. For $y \in \F^k$ define $\Gamma_y:\prod_{i=1}^C \U_{k_i+1} \to [0,1]$ as
$$
\Gamma_y(z_1,\ldots,z_C) = 1 - v(y) - (-1)^{v(y)} \Gamma(z_1,\ldots,z_C).
$$
The probability that $\mu_{f,k}$ samples $v:\F^k \to \{0,1\}$ is
$$
\mu_{f,k}[v] = \E_A \prod_{y \in \F^k} \Gamma_y(P_1(Ay),\ldots,P_C(Ay)).
$$
Expanding each $\Gamma_y$ in the Fourier basis gives
$$
\Gamma_y(z_1,\ldots,z_C) = \sum_{\alpha \in \prod_{i=1}^C \U_{k_i+1}} \hat{\Gamma}_y(\alpha) \expo{\sum_{j=1}^C \alpha_j z_j}.
$$
Note that $\|\widehat{\Gamma_y}\|_{\infty} \le 1$. Plugging this into the equation for $\mu_{f,k}[v]$ and expanding gives
$$
\mu_{f,k}[v] = \sum_{\alpha: \F^k \to \prod_{i=1}^C \U_{k_i+1}} c_{\alpha} \E_A \expo{\sum_{y \in \F^k} \sum_{j=1}^C \alpha(y)_j P_j(Ay)},
$$
where $c_{\alpha} := \prod_{y \in \F^k} \hat{\Gamma}_y(\alpha(y))$. Note that $|c_{\alpha}| \le 1$ and that it depends only on $\Gamma$ and the depths of the polynomials, and not on the specific polynomials. We will apply the same expansion to $g$ and obtain
$$
\mu_{g,k}[v] = \sum_{\alpha: \F^k \to \prod_{i=1}^C \U_{k_i+1}} c_{\alpha} \E_A \expo{\sum_{y \in \F^k} \sum_{j=1}^C \alpha(y)_j Q_j(Ay)}.
$$

We next apply Theorem 3.3 in~\cite{BFHHL13}. It states that linear combination of systems of high rank polynomials evaluated over affine linear forms, such as
$\sum_{y \in \F^k} \sum_{j=1}^C \alpha(y)_j P_j(Ay)$, are either identically zero or very close to uniformly distributed. Concretely, the theorem states that for any parameter $\gamma(C)>0$, if we choose $r(C)$ large enough, then for any polynomials $P_1,\ldots,P_C$ of degrees $d_1,\ldots,d_C$ and depths $k_1,\ldots,k_C$ and rank at least $r(C)$, either
$$
\sum_{y \in \F^k} \sum_{j=1}^C \alpha(y)_j P_j(Ay) \equiv 0
$$
or
$$
\left| \E_A \expo{\sum_{y \in \F^k} \sum_{j=1}^C \alpha(y)_j P_j(Ay)} \right| \le \gamma(C).
$$
Note that we can apply this theorem both to $P_1,\ldots,P_C$ and to $Q_1,\ldots,Q_C$, obtaining the same results. Hence, we conclude that
$$
|\mu_{f,k}[v] - \mu_{g,k}[v]| \le \left(\prod_{i=1}^C p^{k_i+1}\right) \gamma(C) \le p^{dC} \gamma(C).
$$
Choosing $\gamma(C) := p^{-dC} 2^{-p^{k}} \eps$ we obtain that $\mu_{f,k}$ and $\mu_{g,k}$ have statistical distance at most $\eps$.
\end{proof}

\section{Proof of \cref{thm:main}}
To simplify the presentation we prove the theorem for the special case of $R=2$. The proof easily generalizes to $R > 2$, by decomposing every function $f:\F^n \to [R]$ as $f=f^{(1)}+\ldots f^{(R)}$ where $f^{(i)}$ is the indicator function of the set $\{x : f(x)=i\}$.

Let $\cP \subseteq \{f:\F^n \to \{0,1\}: n \in \N\}$ be a testable affine-invariant property, and let $\delta, \eps>0$ be the parameters given in~\cref{thm:main}. We will show that a local test can distinguish between functions which are $\delta$-close to $\cP$ to functions which are $\delta+\eps$ far from $\cP$.

%
Let $m=m(\cP, \delta,\eps)$ be a sufficiently large integer to be determined later, and let $A:\F^{m} \to \F^n$ be a random affine embedding. We will establish the following two statements.
\begin{itemize}
 \item[{\bf (i)}] If $f:\F^n \to \{0,1\}$ is $\delta$-close to $\cP$, then
$$\Pr[\dist(Af,\cP) < \delta+(\eps/2)]>2/3.$$

 \item[{\bf (ii)}] If $f:\F^n \to \{0,1\}$ is $(\delta+\eps)$-far from $\cP$, then
$$\Pr[\dist(Af,\cP) > \delta+(\eps/2)]>2/3.$$
\end{itemize}

\subsection{Proof of (i)}
Since  $f$ is $\delta$-close to $\cP$, there exists a function $g\in \cP$ such that $\alpha:=\|f-g\|_1 \le \delta$. Note that
$$\E_A[\|Af-Ag\|_1]= \E_{x \in \F^m, A}[|Af(x)-Ag(x)|]=  \E_{x \in \F^n}[|f(x)-g(x)|]=\alpha.$$
Furthermore
\begin{eqnarray*}
\E_A\left[\|Af-Ag\|_1^2\right]&=&\E_{x,y \in \F^m, A}[|Af(x)-Ag(x)||Af(y) - Ag(y)|] \le \alpha^2+ \Pr[x=y] \le  \alpha^2 + \frac{1}{p^m}.
\end{eqnarray*}
Hence $\mathbf{Var}[\dist(Af,\cP)] \le \frac{1}{p^{m}}$, and thus by Chebyshev's inequality
$$\Pr[\dist(Af,\cP) \ge \delta+(\eps/2)]< \frac{4}{p^{m} \eps^2}< 1/3,$$
provided that $m$ is sufficiently large.

\subsection{Proof of (ii)}

We apply \cref{cor:testing3} and \cref{cor:close_to_cP_gowers} with parameter $\eps/8$ to obtain $k,d \ge 1$ and $\rho>0$, so that the following two statements hold.
\begin{itemize}
\item If $f:\F^n \to \{0,1\}$, $g:\F^m \to \{0,1\}$ are functions, $g \in \cP$, and $\mu_{f,k}$ and $\mu_{g,k}$ have statistical distance at most $1/4$, then $f$ is $(\eps/8)$-close to $\cP$.
\item If $f,g:\F^n \to \{0,1\}$ are functions and $\|f-g\|_{U^{d}} < \rho$, then $\mu_{f,k}$ and $\mu_{g,k}$ have statistical distance at most $1/100$.
\end{itemize}

Let $f:\F^n \to \{0,1\}$ be a function which is $(\delta+\eps)$-far from $\cP$. We start by decomposing $f$ to a structured part and a pseudo-random part.
Our decompositions will use a number of parameters. We already fixed $\eps,\delta$ and $d$. Let $\gamma>0$, $\eta_0, \eta_1: \N \to \R^{+}$ and $r_0, r_1: \N \to \N$ be parameters to be determined later. For the reader who wishes to verify that these definitions are not cyclical, we note that $\gamma$ will depend just on $\eps,\delta,d$; that $\eta_1,r_1$ will depend just on $\eps,\delta,d,\gamma$; and that $\eta_0,r_0$ will depend on $\eps,\delta,d,\gamma,\eta_1,r_1$.

We apply \cref{thm:reg} to $f$ with parameters $d, \gamma, r_0, \eta_0$ and a trivial initial factor, and obtain an $r_0$-regular polynomial factor $\cB_0$ of degree less than $d$, and a decomposition
$$
f=f_1+f_2+f_3,
$$
where $f_1=\E[f|\cB_0],\|f_2\|_2 < \gamma,\|f_3\|_{U^d} < \eta_0(|\cB_0|)$. Next, we project this decomposition to $Af$. Suppose that $\cB_0$ is defined by polynomials $P_1,\ldots,P_C$. Denote by $Q_i:=AP_i$ for $i=1,\ldots,C$, and let $\tilde{\cB}_0$ be the polynomial factor over $\F^m$ defined by the $Q_i$'s. We decompose
$$
Af = Af_1 + Af_2 + Af_3.
$$

The following claim argues that $Af_1,Af_2,Af_3$ have similar properties to $f_1,f_2,f_3$ with high probability, assuming that $m$ is chosen large enough and that $r_0$ is chosen to grow fast enough.

\begin{claim}
\label{claim:1}
Assume that $r_0(C) \ge r_{\ref{rankreg}}(d,1/(2 p^{dC}))$.
If $m \ge m_{\ref{claim:1}}(d,\gamma,r_0,\eta_0)$, then the following events hold with probability at least $99\%$.
\begin{itemize}
 \item[$(\mathbf{E}_1)$] The polynomials $Q_1,\ldots,Q_C$ have the same degrees and depths as $P_1,\ldots,P_C$, respectively, and $\tilde{\cB}_0$ is $r_0$-regular.

 \item[$(\mathbf{E}_2)$] We have $\left\|Af_2\right\|_2 \le 2\gamma$ and $\left\|Af_3\right\|_{U^d} \le 2\eta_0(|\cB_0|)$.

 \item[$(\mathbf{E}_3)$]  $\left\|\E[Af| \tilde{\cB}_0]-A\E[f|\cB_0]\right\|_\infty \le \gamma$.

\end{itemize}
\end{claim}
\begin{proof}
\cref{lem:rankrestrict} shows that the probability that  $(\mathbf{E}_1)$ does not hold can be made arbitrarily small by setting the parameters properly.
To prove that $(\mathbf{E}_2)$ holds with high probability, we will show that
 $\Pr\left[\left\|Af_3\right\|_{U^d} > \|f_3\|_{U^d}+o_{m}(1) \right]=o_{m}(1)$ and
$\Pr\left[\|Af_2\|_2 > \|f\|_2 + o_{m}(1) \right]=o_{m}(1)$. We only prove the former, as the proof of the latter is easy and similar.
To do so, We will establish that $\E_A \left\|Af_3\right\|_{U^d}^{2^{d}} = \left\|f_3\right\|_{U^d}^{2^{d}}+o_{m}(1)$ and that $\mathbf{Var}(\left\|Af_3\right\|_{U^d}^{2^{d}})=o_{m}(1)$, and apply Chebyshev's inequality.

We first establish the first moment calculation. Let $y_1,\ldots,y_{d},x$ be uniform random variables taking values in $\F^m$. Note that if $y_1,\ldots,y_d,x$ are linearly independent then
 $Ay_1,\ldots,Ay_{d},Ax$ are linearly independent uniform random variables taking values in $\F^n$. The probability that they are not linearly
 independent is at most $p^{d+1-m}=o_{m}(1)$. Combining this with the fact that $\|f_3\|_{\infty} \le 1$ we get
$$
\E_A  \left\|Af_3\right\|_{U^d}^{2^{d}}  = \Ex{A \Delta_{y_1,\ldots,y_{d}}f_3(x)}= \Ex{\Delta_{Ay_1,\ldots,Ay_{d}}f_3(Ax)} = \|f_3\|_{U^d}^{2^{d}} \pm p^{d+1-m}.$$

We proceed to the second moment calculation. Let $y'_1,\ldots,y'_{d},x'$ be independent uniform random variables taking values in $\F^m$ uniformly and independently of $y_1,\ldots,y_{d},x$.  Similarly, if $y_1,\ldots,y_d,x,y'_1,\ldots,y'_d,x'$ are linearly independent then $Ay_1,\ldots,Ay_d,Ax,Ay'_1,\ldots,Ay'_d,Ax'$ are linearly independent uniform random variables taking values in $\F^n$. Hence same as before,
\begin{eqnarray*}
\E_A  \left\|Af_3\right\|_{U^d}^{2^{d+1}}  &=& \Ex{\left(A \Delta_{y_1,\ldots,y_{d}}f_3(x)\right) \left(A\Delta_{y'_1,\ldots,y'_{d}}f_3(x') \right)}
= \|f_3\|_{U^d}^{2^{d+1}} \pm p^{2d+2-m}.
\end{eqnarray*}
Thus, $\mathbf{Var}(\left\|Af_3\right\|_{U^d}^{2^{d}}) = o_{m}(1)$ and by Chebyshev's inequality $(\mathbf{E}_2)$ holds with high probability assuming $m$ is chosen large enough.

We next establish that $(\mathbf{E}_3)$ holds with high probability. Similarly to the previous calculation, this will also be shown by performing a first and second moment calculation and applying Chebyshev's inequality.
Consider an atom $b_0 \in \T^C$ of $\cB_0$.  Since $\tilde{\cB}_0$ is defined by $AP_1,\ldots,AP_C$, we have that $\E[Af(y) | \tilde{\cB}_0(y)=b_0] = \E[Af(y) | \cB_0(Ay)=b_0]$, hence
$$\E_A[\E[Af(y) | \cB_0(Ay)=b_0 ]]=\E_{x \in \F^n}[f(x) | \cB_0(x)= b_0],$$
and
\begin{eqnarray*}
\E_A\left[\E[Af(y) | \cB_0(Ay)=b_0 ]^2\right]&=&\E_{y_1,y_2 \in \F^m, A}[f(Ay_1)f(Ay_2) |\cB_0(Ay_1)=\cB_0(Ay_2)= b_0] \\ &=& \E[f(x) | \cB_0(x)= b_0]^2 \pm \Pr[y_1=y_2|\cB_0(Ay_1)=\cB_0(Ay_2)= b_0]\\
&=&  \E[f(x) | \cB_0(x)= b_0]^2 \pm p^{-m} / (|\{x \in \F^m: \cB_0(x)=b_0\}| p^{-n})^2 \\
&=&  \E[f(x) | \cB_0(x)= b_0]^2 \pm 4|\cB_0|^2 p^{-m},
\end{eqnarray*}
where in the last step we applied \cref{rankreg} and the assumption on the rank of $\cB_0$ to lower bound the size of the atom defined by $b_0$.
So by Chebyshev's inequality $\E[Af(y) | \tilde{\cB}_0(y)=b_0]$ is concentrated around $\E_{x}[f(x) | \cB_0(x)= b_0]$. Since the number of atoms is bounded by $\|\cB_0\|$, we obtain that with probability $1-o_{m}(1)$ this holds of every atom.
\end{proof}

\cref{claim:1} shows that $99\%$ of the affine embeddings $A$ satisfy $(\mathbf{E}_1),(\mathbf{E}_2),(\mathbf{E}_3)$. Let us assume towards contradiction that $\Pr_A[\dist(Af,\cP) > \delta+(\eps/2)] \le 1/3$. So, we can fix an embedding $A$ so that $(\mathbf{E}_1),(\mathbf{E}_2),(\mathbf{E}_3)$ hold, and find a function $h:\F^m \to \{0,1\}$ in $\cP$ for which $\dist(Af,h) \le \delta+(\eps/2)$. We fix $A$ and $h$ for the reminder of the proof.

The next step is to decompose $h$. However, we wish to maintain the regular factor $\tilde{\cB}_0$ we obtained for $Af$. So, we apply \cref{thm:reg} to $h$ with parameters $d, \gamma, r_1, \eta_1$ and initial factor $\tilde{\cB}_0$, and obtain an $r_1$-regular polynomial factor $\tilde{\cB}_1$ of degree $d$, and a decomposition
$$
h=h_1+h_2+h_3,
$$
where $h_1=\E[h|\tilde{\cB}_1],\|h_2\|_2 < \gamma,\|h_3\|_{U^d} < \eta_1(|\cB_1|)$. Furthermore, we will assume that $r_0$ is much larger than $r_1$, so that by \cref{thm:reg} we get that $\tilde{\cB}_1$ is a syntactic refinement of $\tilde{\cB}_0$ (which we recall that by \cref{claim:1} is $r_0$-regular).
Concretely, this will require us to assume that $r_0(C) \ge r_{\ref{thm:reg}}(\eta_1,  r_1, \gamma, C, d)$. So, $\tilde{\cB}_1$ is defined by polynomials $Q_1,\ldots,Q_{C'}$ for a constant $C' > C$, where we recall that $\tilde{\cB}_0$ was defined by $Q_1,\ldots,Q_C$. Our construction so far guarantees that for $i \le C$ we have that $Q_i = A P_i$ and that $P_i,Q_i$ have the same depth and degree. We would like to guarantee this also for $i>C$. That is, we would like to find polynomials $P_i$ for $C<i \le C'$ defined over $\F^n$ for which $Q_i = AP_i$ as well, and such that $P_1,\ldots,P_{C'}$ is of high rank.

Let $A':\F^n \to \F^m$ be any affine transformation satisfying $A A'=I_m$. For $i=C+1,\ldots,C'$ set $P_{i} := A' Q_{i}$. Note that applying an affine transformation cannot increase degree, depth, or rank. Hence, since $P_i = A' Q_i$ and $Q_i = A P_i$ we must have that $P_i,Q_i$ have the same degree and depth for $C<i \le C'$. Moreover, by $(\mathbf{E}_1)$ we know that  this also holds for $i \le C$, hence it holds for all $1 \le i \le C'$. Furthermore, since by construction $Q_1,\ldots,Q_{C'}$ are $r_1$-regular than so are $P_1,\ldots,P_{C'}$. We denote by $\cB_1$ the polynomial factor defined by $P_1,\ldots,P_{C'}$. The following claim summarizes the notation and the facts we established so far.

\begin{claim}\label{claim:2}
The factor $\cB_0$ is an $r_0$-regular factor over $\F^n$ defined by $P_1,\ldots,P_C$.  The factor $\tilde{\cB}_0$ is an $r_0$-regular factor over $\F^m$ defined by $Q_1,\ldots,Q_C$. The factor $\cB_1$ is an $r_1$-regular factor over $\F^n$ defined by $P_1,\ldots,P_{C'}$.  The factor $\tilde{\cB}_1$ is an $r_1$-regular factor over $\F^m$ defined by $Q_1,\ldots,Q_{C'}$. We further have:
\begin{itemize}
\item $Q_i = AP_i$.
\item $P_i,Q_i$ have the same degree and depth.
\item $f=f_1+f_2+f_3$ where $f_1 = \E[f|\cB_0], \|f_2\|_2 \le \gamma, \|f_3\|_{U^d} \le \eta_0(C)$.
\item $h:\F^m \to \{0,1\}$ is a function in $\cP$ to $Af$ for which $\dist(Af,h) \le \delta+\eps/2$.
\item $h=h_1+h_2+h_3$ where $h_1 = \E[h|\tilde{\cB}_1], \|h_2\|_2 \le \gamma, \|h_3\|_{U^d} \le \eta_1(C')$.
\end{itemize}
\end{claim}

We already know that the property that $h$ is in $\cP$ does not actually depend on $h_3$. We will shortly show that for a small enough choice of $\gamma$, it also does not depend on $h_2$, hence all the information is essentially just in $h_1$. We would like to lift $h_1$ from $\F^m$ to $\F^n$ in order to get a function $g:\F^n \to \{0,1\}$ which is close to $\cP$. Moreover, we would like to do so in a way so that $\|f-g\|_1 \approx \dist(Af,h)$, hence this lifting must be done carefully. We start by lifting $h_1:\F^m \to [0,1]$ to a function $\phi:\F^n \to [0,1]$ which has a similar structure. We would later use $\phi$ to define the required function $g$ as discussed above.

The two factors $\cB_1$ and $\tilde{\cB}_1$ are both of large rank (at least $r_1(C')$) and their defining polynomials are in a degree and depth preserving one-to-one correspondence. This naturally defines an operator $\cT$ that maps functions $\phi:\F^n \to \C$ measurable with respect to $\cB_1$ to functions $\cT \phi:\F^m \to \C$ measurable with respect to $\tilde{\cB}_1$. More precisely, $\cT:\Gamma(P_1,\ldots,P_{C'}) \mapsto \Gamma(Q_1,\ldots,Q_{C'})$ for every function $\Gamma$. Note that $\cT$ is invertible. So, recalling that $h_1=\F^m \to [0,1]$ is measurable with respect to $\tilde{\cB}_1$, it is natural to define $\phi:\F^n \to [0,1]$ by
$$
\phi := \cT^{-1} h_1.
$$

We first argue that $\phi$ is close to $\E[f|\cB_1]$ assuming the underlying parameters are properly chosen.

\begin{claim}\label{claim:3}
Assume that $\eta_0(C) \le p^{-d C_{\ref{thm:reg}}(\eta_1,  r_1, \gamma, C, d)} \gamma$ and $r_1(C) \ge r_{\ref{rankreg}}(d,p^{-d C} \gamma)$. Then
$$
\|\E[f|\cB_1] - \phi\|_1 \le \delta + \eps/2 + 9 \gamma.
$$
\end{claim}

\begin{proof}
By \cref{atomsize} and the condition on $r_1$, every function $\phi:\F^n \to [-1,1]$ satisfies
\begin{equation}
\label{eq:transfer}
|\|\cT \phi \|_1 - \|\phi\|_1 | \le \gamma.
\end{equation}
We can thus write
\begin{eqnarray}
\|\E[f|\cB_1] - \cT^{-1} \E[h|\tilde{\cB}_1]\|_1 &\le&  \|\cT \E[f|\cB_1]  -  \E[h|\tilde{\cB}_1] \|_1 + \gamma \nonumber \\
&\le &  \|\cT \E[f|\cB_1]  -  \cT \E[f|\cB_0] \|_1 + \|\cT \E[f|\cB_0]  -  \E[Af|\tilde{\cB_0}] \|_1 +  \nonumber \\
 && +  \|\E[Af|\tilde{\cB}_0]  -  \E[Af|\tilde{\cB}_1] \|_1+ \|\E[Af | \tilde{\cB}_1] - \E[h | \tilde{\cB}_1]\|_1 + \gamma.  \label{eq:main}
\end{eqnarray}
We will show that each of the terms is bounded by $O(\gamma)$, except for the fourth term for which
$$
\|\E[Af | \tilde{\cB}_1] - \E[h | \tilde{\cB}_1]\|_1 \le \|Af-h\|_1 \le \delta + \eps/2.
$$
To bound the first term, we apply $\eqref{eq:transfer}$ and \cref{lem:refinement},
$$
\|\cT \E[f|\cB_1]  -  \cT \E[f|\cB_0] \|_1 \le \|\E[f|\cB_1]  -  \E[f|\cB_0] \|_1 + \gamma \le \|f_2\|_2 + p^{dC'} \|f_3\|_{U^d}+\gamma \le 3 \gamma.
$$
To bound the second term, by ($\mathbf{E}_3$)
$$
\|\cT \E[f|\cB_0]  -  \E[Af|\tilde{\cB}_0] \|_1 = \|A \E[f|\cB_0]  -  \E[Af|\tilde{\cB_0}] \|_1 \le \|A \E[f|\cB_0]  -  \E[Af|\tilde{\cB_0}] \|_{\infty} \le \gamma.
$$
To bound the third term, by \cref{lem:refinement} and ($\mathbf{E}_2$),
$$
\|\E[Af|\tilde{\cB}_0]  -  \E[Af|\tilde{\cB}_1] \|_1  \le \|Af_2\|_2 + p^{dC'}\|Af_3\|_{U^d} \le 2\gamma+2\gamma \le 4\gamma.
$$
\end{proof}

The function $\phi$ is defined over $\F^n$ and has the same structure as $h_1$, however we only guaranteed that it is close to $\E[f|\cB_1]$. This cannot be avoided since $\phi$ is $\cB_1$-measurable, and thus it cannot approximate $f$ inside the atoms. The next step is to define a function whose average on atoms is $\phi$, and that simultaneously has a small distance from $f$. It will be obtained by perturbing $f$ inside each atom in order to obtain the right average while making the changes minimal. It will be convenient to first define such a function $\psi$ mapping $\F^n$ to $[0,1]$. Later we will use it to sample $g:\F^n \to \{0,1\}$ with the required properties.

Define a function $\psi:\F^n \to [0,1]$ as follows. Fix $x \in \F^n$ and let $b_1=\cB_1(x)$ be its corresponding atom in $\cB_1$. Let $\alpha = \E[f|\cB_1](b_1)$ be the average of $f$ over the atom and $\beta=\phi(b_1)$ be the value $\phi$ attains on the atom. We set
$$
\psi(x) = \left\{
\begin{array}{cc}
\frac{\beta-\alpha}{1-\alpha}&\textrm{If } \alpha \le \beta \textrm{ and } f(x)=0\\
1&\textrm{If } \alpha \le \beta \textrm{ and } f(x)=1\\
0&\textrm{If } \alpha > \beta \textrm{ and } f(x)=0\\
\frac{\beta}{\alpha}&\textrm{If } \alpha > \beta \textrm{ and } f(x)=1
\end{array}
\right.
$$

The following claim summarizes the properties of $\psi$, assuming the underlying parameters are properly chosen.

\begin{claim}\label{claim:4}
Assume that $\eta_0(C) \le p^{-d C_{\ref{thm:reg}}(\eta_1,  r_1, \gamma, C, d)} \gamma$. Then
\begin{itemize}
\item[(i)] $\E[\psi|\cB_1] = \phi$.
\item[(ii)] $\|f-\psi\|_1 = \|\E[f|\cB_1] - \phi\|_1$.
\item[(iii)] $\|\psi - \phi\|_{U^d} \le \gamma + 3 \gamma^{1/2^{d}}$.
\end{itemize}
\end{claim}

\begin{proof}
Items $(i)$,$(ii)$ follow immediately from the definition of $\psi$. To establish $(iii)$, decompose $\psi-\phi$ to the atoms of $\cB_1$,
$$
\psi(x) - \phi(x) = \sum_{b_1 \in \cB_1} \left(\psi(x) - \phi(b_1)\right) 1_{\cB_1(x)=b_1}.
$$
Fix an atom $b_1$ of $\cB_1$ and set $\alpha=\E[f|\cB_1](b_1)$ and $\beta=\phi(b_1)$. If $\alpha \le \beta$, then
$$
\left(\psi(x) - \beta\right)1_{\cB_1(x)=b_1} = \left(f(x) + \frac{\beta-\alpha}{1-\alpha}(1-f(x)) - \beta\right) 1_{\cB_1(x)=b_1} = \frac{1-\beta}{1-\alpha}(f(x)-\alpha) 1_{\cB_1(x)=b_1}.
$$
If $\alpha > \beta$ then
$$
\left(\psi(x) - \beta\right)1_{\cB_1(x)=b_1} = \left(\frac{\beta}{\alpha} f(x) - \beta\right) 1_{\cB_1(x)=b_1} = \frac{\beta}{\alpha}(f(x)-\alpha) 1_{\cB_1(x)=b_1}.
$$
Recall that $f=f_1+f_2+f_3$ with $f_1=\E[f|\cB_0], \|f\|_2 \le \gamma, \|f_3\|_{U^d} \le \eta_0(|\cB_0|)$. Define $f'_i = f_i - \E[f_i|\cB_1]$. Note that
$f(x)-\E[f|\cB_1](x) = f'_1(x)+f'_2(x)+f'_3(x)$. Furthermore, $f'_1=0$ since $\cB_1$ is a refinement of $\cB_0$. Define $\xi(x)=(1-\beta)/(1-\alpha)$ if $x$ belongs to an atom $b_1$ with $\alpha \le \beta$, and $\xi(x)=\beta/\alpha$ otherwise. We thus get that
$$
\psi(x) - \phi(x) = \sum_{b \in \cB_1} \xi(x) (f'_2(x)+f'_3(x)) 1_{\cB_1(x)=b_1} = \xi(x) (f'_2(x)+f'_3(x)).
$$
We now bound $\|\psi-\phi\|_{U^d}$.  Now, by \cref{cl:gowers_l1} and the fact that $\|\xi\|_{\infty} \le 1$,
$$
\|\xi \cdot f'_2\|_{U^d}^{2^{d}} \le \|\xi \cdot f'_2\|_1 \le \|f_2\|_1+\|\E[f_2|\cB_1]\|_1 \le 2 \|f_2\|_1 \le 2 \|f_2\|_2 \le 2 \gamma.
$$
By \cref{cl:L2U} , we have
$$
\|\xi \cdot \E[f_3|\cB_1]\|_{U^d}^{2^{d}} \le \|\xi \cdot \E[f_3|\cB_1]\|_1 \le \|\E[f_3|\cB_1]\|_1 \le p^{d|\cB_1|}  \|f_3\|_{U^d} \le p^{d|\cB_1|}  \eta_0(|\cB_0|)  \le \gamma.
$$
Thus recalling that  $\|\cB_1\|$ is an upper-bound on the number of atoms of $\cB_1$, by \cref{cl:restrict_atom} we have
\begin{eqnarray*}
\|\xi \cdot f'_3\|_{U^d} &\le& \|\xi \cdot E[f_3|\cB_1]\|_{U^d} + \|\xi \cdot f_3\|_{U^d} \le \gamma^{1/2^d}+ \sum_{b_1 \in \cB_1} \| f_3 (x) 1_{\cB_1(x)=b_1}\|_{U^d} \\ &\le& \gamma^{1/2^d} + \|\cB_1\| \|f_3\|_{U^d}
\le \gamma^{1/2^d} + \|\cB_1\| \eta_0(|\cB_0|)  \le  \gamma^{1/2^d} + \gamma .
\end{eqnarray*}
\end{proof}

We are ready to conclude the proof. Let $g:\F^n \to \{0,1\}$ be sampled with $\Pr[g(x)=1]=\psi(x)$ independently for all $x \in \F^n$. The following simple claim states that with high probability, $g$ behaves like $\psi$.

\begin{claim}\label{claim:5}
If $n$ is large enough then with probability at least $99\%$ over the choice of $g$,
\begin{itemize}
\item $\|f-g\|_1 \le \|f-\psi\|_1 + \gamma$.
\item $\|g-\psi\|_{U^d} \le \gamma$.
\end{itemize}
\end{claim}

\begin{proof}
Both items are simple. The first item holds since $\E_g \|f-g\|_1 = \|f-\psi\|_1$, and $\Var \|f-g\|_1 \le p^{-n}$. Hence by Chebyshev's inequality, $\|f-g\|_1 = \|f-\psi\|_1 + \gamma$ with probability $1-o_n(1)$. For the second item,
$$
\E_g \|g - \psi\|_{U^d}^{2^d} = \E_{x,y_1,\ldots,y_d \in \F^n} \prod_{I \subseteq [d]} \E_g[g(x+\sum_{i \in I} y_i) - \psi(x+\sum_{i \in I} y_i)].
$$
The inner expectation is zero if all the sums $x+\sum_{i \in I} y_i$ are distinct. In particular this holds if $y_1,\ldots,y_d$ are linearly independent, which happens with probability at least $1-p^{d-n}$. Hence
$$
\E_g \|g - \psi\|_{U^d}^{2^d} \le p^{d-n} = o_n(1),
$$
and so by choosing $n$ large enough we get that with high probability $\|g - \psi\|_{U^d} \le \gamma$.
\end{proof}

Fix $g:\F^n \to \{0,1\}$ satisfying \cref{claim:5}. Let us summarize the facts that we know so far. By \cref{claim:3}, \cref{claim:4} and \cref{claim:5} we know that $f$ and $g$ are close. By choosing $\gamma \le \eps/40$ we get that
$$
\|f-g\|_1 \le \delta + \eps/2 + 10 \gamma \le \delta + 3 \eps/4.
$$
We next argue that $g$ is very close to $\cP$. Recall that $h \in \cP$ and $h_1=\E[h|\tilde{\cB}_1]$. We have that
$$
\|h - h_1\|_{U^{d}} \le \|h_2\|_{U^d} + \|h_3\|_{U^d} \le \|h_2\|_{2}^{1/2^d} + \|h_3\|_{U^d} \le \gamma^{1/2^d} + \eta_1(C'),
$$
where we applied \cref{cl:gowers_l1} and the assumptions on $h_2,h_3$. By choosing $\gamma \le (\rho/2)^{2^d}$ and $\eta_1(C') \le \rho/2$, we get that $\|h-h_1\|_{U^d} \le \rho$. Hence by \cref{lemma:close_sampling_gowers} the statistical distance between $\mu_{h,k}$ and $\mu_{h_1,k}$ is at most $1/100$. Next, by definition $h_1=\Gamma(Q_1,\ldots,Q_{C'})$ and $\phi=\Gamma(P_1,\ldots,P_{C'})$, where both $Q_1,\ldots,Q_C$ and $P_1,\ldots,P_C$ are $r_1$-regular. Making sure that $r_1(C') \ge r_{\ref{lemma:mu_regular}}(C)$ guarantees that conditions of \cref{lemma:mu_regular} are satisfied, and hence $\mu_{h_1,k}$ and $\mu_{\phi,k}$ also have statistical distance at most $1/100$. Next, by \cref{claim:4} and \cref{claim:5} we have
$$
\|g-\phi\|_{U^{d}} \le 2 \gamma + 3 \gamma^{1/2^d}.
$$
Choosing $\gamma$ so that $2 \gamma + 3 \gamma^{1/2^d} \le \rho$, we get that also the statistical distance between $\mu_{g,k}$ and $\mu_{\phi,k}$ is at most $1/100$. Hence, combining all these together gives that
$$
\textrm{The statistical distance between } \mu_{g,k} \textrm{ and } \mu_{h,k} \textrm{ is at most 3/100}.
$$
Applying \cref{cor:close_to_cP_gowers} we establish that
$$
\dist(g,\cP) \le \eps/8.
$$
Hence we reached a contradiction, since $f$ is $(\delta+3 \eps/4)$-close to a function $g$ which is $(\eps/8)$-close to $\cP$, that is $f$ is $(\delta + 7 \eps/8)$-close to $\cP$, contradicting our initial assumption that $f$ is $(\delta+\eps)$-far from $\cP$.

\bibliographystyle{alpha}
\bibliography{testing}

\end{document}